\documentclass[pra,
twocolumn, 
superscriptaddress,
showpacs,preprintnumbers,amsmath,amssymb]{revtex4}

\usepackage{graphicx}
\usepackage[caption=false]{subfig}
\usepackage{booktabs}
\usepackage{diagbox}
\usepackage{amsthm}
\usepackage{tensor}
\usepackage{color}
\usepackage[all]{xy}
\usepackage{tikz}
\usepackage{dsfont}
\usepackage{times,txfonts}
\usetikzlibrary{positioning}
\usepackage{braket}
\usepackage{mathtools}
\usepackage{amsmath, amssymb}
\usepackage[T1]{fontenc}
\usepackage{enumitem}
\usepackage{hyperref}
\setlist[enumerate]{label=(\roman*),align=left,labelsep=.5em,leftmargin=*,widest=ix,itemsep=0pt,parsep=0pt,topsep=.3\baselineskip}

\newtheorem{Thm}{Theorem}

\newtheorem{Prop}[Thm]{Proposition}
\newtheorem{Cor}[Thm]{Corollary}
\theoremstyle{definition}

\newtheorem{Ex}[Thm]{Example}

\newcommand{\rank}{\operatorname{rank}}
\newcommand{\Tr}{\operatorname{Tr}}
\newcommand{\id}{\mathbb{I}}
\newcommand{\sep}{\textit{SEP}}
\newcommand{\ppt}{\textit{PPT}}
\newcommand{\red}{\textit{RED}}
\newcommand{\maj}{\textit{MAJ}}
\newcommand{\cond}{\textit{CEN}}
\newcommand{\ut}{\textit{UND}}
\newcommand{\uo}{\textit{UND}_{\rightarrow}}
\newcommand{\condl}{\textit{CEN}_{\!\leftarrow}}
\newcommand{\condrr}{\textit{CEN}_{\!\rightarrow}}

\begin{document}

\title{Maximal extension on converse monogamy of entanglement for tripartite pure states}

\author{Junhyeong An}\email{hsj0356@khu.ac.kr}
\affiliation{Department of Mathematics and Research Institute for Basic Sciences, Kyung Hee University, Seoul 02447, Korea}

\author{Soojoon Lee}\email{level@khu.ac.kr}
\affiliation{Department of Mathematics and Research Institute for Basic Sciences, Kyung Hee University, Seoul 02447, Korea}
\affiliation{School of Computational Sciences, Korea Institute for Advanced Study, Seoul 02455, Korea}

\date{\today}

\begin{abstract}
Unlike classical correlations, entanglement cannot be freely shared among multiple parties.
This unique feature of quantum systems is known as the monogamy of entanglement.
While it holds for all multipartite pure states,  its converse --- weak entanglement between two parties enforces strong entanglement with a third party --- occurs only under specific conditions.
In particular, Hayashi and Chen [Phys. Rev. A \textbf{84}, 012325 (2011)] demonstrated a qualitative version of the converse monogamy of entanglement (CMoE) for tripartite pure states by employing a hierarchy of bipartite entanglement defined through the relations among various separability criteria, and Singh and Datta [IEEE Trans. Inf. Theory \textbf{69}, 6564 (2023)] later extended this notion of the CMoE from the viewpoint of distillability under one-way or two-way classical communication.
In this work, we extend their results to the CMoE with broader conditions, and furthermore show that our extensions are maximal with respect to the hierarchies they considered.
\end{abstract}

\maketitle

\section{Introduction}
Quantum entanglement plays a central role in quantum information theory. 
It serves as a fundamental resource for many quantum information processing tasks, including quantum teleportation~\cite{bennett1993teleporting}, superdense coding~\cite{bennett1992communication}, and quantum cryptography~\cite{bennett2014quantum}.
To utilize these applications effectively, it is essential to understand how entanglement is distributed among multiple parties.

The distribution of entanglement among multiple parties is subject to intrinsic restrictions that have no classical counterpart.
This restriction is known as the \textit{monogamy of entanglement}~\cite{coffman2000distributed,terhal2004entanglement}. 
For instance, in a multipartite system, if a pair of parties is maximally entangled, then neither can share entanglement with any part of the remaining system. 
In particular, for a tripartite pure state shared by Alice, Bob, and Charlie, if Alice and Bob are strongly entangled, they must be weakly entangled with Charlie.

On the other hand, the converse statement --- if a bipartite marginal state of a tripartite pure quantum state is weakly entangled, the other ones must be strongly entangled --- does not generally hold. 
For example, the GHZ state provides a clear counterexample, since all of its two-qubit reduced states are separable. 
Nevertheless, by imposing suitable conditions on a reduced bipartite state, one can meaningfully formulate the \emph{converse monogamy of entanglement} (CMoE) within tripartite pure states.  
Hayashi and Chen~\cite{hayashi2011weaker} introduced the CMoE based on a hierarchy obtained from several separability criteria, and later Singh and Datta~\cite{singh2023fully} proposed a similar formulation from the perspective of distillability.
Although these two studies have approached the CMoE from different viewpoints, the structural connection between them or the maximal range in which the CMoE is valid have not yet been fully clarified.

In this work, we reconstruct the two formulations of the CMoE within extension hierarchical framework.
Our results progressively relax the conditions introduced in the previous formulations, and provide explicit examples that cannot be explained by earlier results. 
Furthermore, by identifying examples in which the CMoE no longer holds, we clarify the boundary beyond which the CMoE fails, thereby demonstrating that our two formulations presented here constitute the maximal extensions of the CMoE.

The remainder of this paper is organized as follows. 
In Sec.~\ref{sec:preli}, we review the separability criteria we here consider, two slightly different hierarchies with respect to the criteria, and the previous results on the CMoE. 
In Sec.~\ref{sec:main}, we present two forms of generalizations of the CMoE and prove, through explicit examples, that these generalizations constitute its maximal extensions. 
Finally, in Sec.~\ref{sec:con}, we summarize our conclusions and discuss possible directions for future work.

\section{Preliminaries}\label{sec:preli}

To rigorously define the CMoE, one requires a hierarchy of bipartite entanglement that characterizes what it means for entanglement between two parties to be \textit{weak} or \textit{strong}. 
In this section, we introduce several separability criteria that are necessary for defining the CMoE and summarize their hierarchical relationships. 
This hierarchy serves as basic knowledge for understanding the two forms of the CMoE, which will be reviewed in the following subsections. 

\subsection{Separability criteria and hierarchies of bipartite entanglement}\label{subsec:def}

The nine classes associated with separability criteria are defined as follows.  
Let $\rho_{AB}$ be a bipartite quantum state acting on the Hilbert space
$\mathcal{H}_A \otimes \mathcal{H}_B$, and let $\rho_A = \Tr_B(\rho_{AB})$ and $\rho_B = \Tr_A(\rho_{AB})$ denote its reduced states.

\begin{enumerate}
    \item $\sep$: $\rho_{AB} \in \sep$ if and only if $\rho_{AB} = \sum_k p_k\, \rho_k^A \otimes \rho_k^B$.
    \item $\ppt$: $\rho_{AB} \in \ppt$ if and only if $\rho_{AB}^\Gamma \ge 0$, where $\Gamma$ denotes partial transposition.
    \item $\ut$: $\rho_{AB} \in \ut$ if and only if $\rho_{AB}$ is \emph{undistillable}, i.e., there exists no local operations and classical communication (LOCC) protocol converting $\rho_{AB}^{\otimes n}$ into $\ket{\phi^+}_{AB}$ with a nonvanishing asymptotic yield as $n \to \infty$. 
    If $\rho_{AB} \notin \ut$, then $\rho_{AB}$ is called \emph{distillable}.
    \item $\uo$: $\rho_{AB} \in \uo$ if and only if $\rho_{AB}$ is \emph{one-way undistillable with respect to $A \to B$}, 
    i.e., no one-way LOCC protocol can asymptotically produce $\ket{\phi^+}_{AB}$.
    \item $\red$: $\rho_{AB} \in \red$ if and only if $\rho_A \otimes I_B \ge \rho_{AB}$ and $I_A \otimes \rho_B \ge \rho_{AB}$.
    \item $\maj$: $\rho_{AB} \in \maj$ if and only if $\rho_A \succ \rho_{AB}$ and $\rho_B \succ \rho_{AB}$, 
    where $\succ$ denotes the majorization relation.
    \item $\condl$: $\rho_{AB} \in \condl$ if and only if $S(\rho_{AB}) \ge S(\rho_A)$, where $S$ is the von Neumann entropy.
    \item $\condrr$: $\rho_{AB} \in \condrr$ if and only if $S(\rho_{AB}) \ge  S(\rho_B)$. 
    \item $\cond$: $\rho_{AB} \in \cond$ if and only if $\rho_{AB} \in \condl \cap \condrr$.
\end{enumerate}

\begin{figure}
	\centering
	\includegraphics[width=1.0\linewidth]{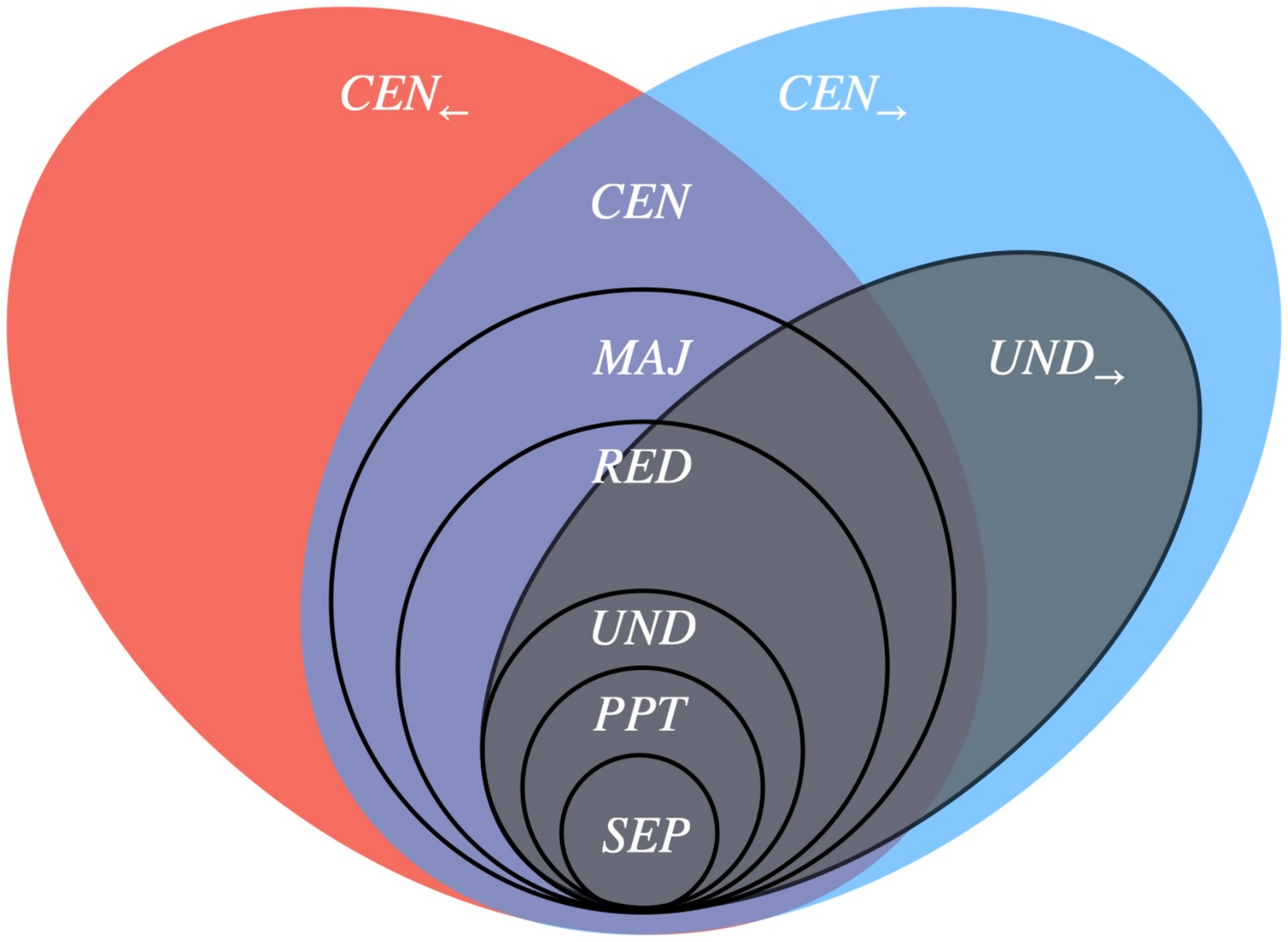}
	\caption{Hierarchy of bipartite states with respect to separability criteria.}
	\label{Fig:Hierarchy}
\end{figure}

From these definitions, the following implication chains hold~\cite{horodecki1996necessary,horodecki1998mixed,horodecki1999reduction,hiroshima2003majorization}, which is depicted in Fig.~\ref{Fig:Hierarchy}:
\begin{align}
&(i)\Rightarrow(ii)\Rightarrow(iii)\Rightarrow(v)\Rightarrow(vi)\Rightarrow(ix),
\label{eq:hierarchy1}\\
&(i)\Rightarrow(ii)\Rightarrow(iii)\Rightarrow(iv)\Rightarrow(viii).
\label{eq:hierarchy2}
\end{align}
In particular, the converse of $(ii)\Rightarrow(iii)$ remains an open problem~\cite{horodecki1998mixed,divincenzo2000evidence}, 
and the other reverse implications do not generally hold~\cite{horodecki1997separability,divincenzo2000evidence,nielsen2001separable,bennett1996mixed}.

This hierarchy may reflect the relative strength of entanglement: for two bipartite quantum states $\rho$ and $\sigma$, if $\rho$ is in one class and $\sigma$ is not, then we can consider $\rho$ to have weaker entanglement than $\sigma$.
Therefore, the hierarchy provides a natural framework for comparing the strength of bipartite entanglement. 
In the context of the CMoE, it asserts that if one pair of subsystems in a tripartite quantum system exhibits weak entanglement, the complementary pair must exhibit strong entanglement, as described in the next two subsections.

\subsection{Converse monogamy of entanglement based on qualitative hierarchy}

We begin with the result of Hayashi and Chen~\cite{hayashi2011weaker}, who first formalized the notion of the CMoE for tripartite pure states using the above hierarchy of separability criteria.  
Their result shows that a strong form of hierarchical collapse occurs 
when one of the reduced states belongs to $\ppt$, as follows. 

\begin{Prop}\label{thm:HC}
Let $\ket{\psi}_{ABC}$ be a tripartite pure state with the reduced state $\rho_{AC} \in \ppt$. 
Then, for the reduced state $\rho_{AB}$, the following are equivalent:
\begin{enumerate}
    \item $\rho_{AB} \in \sep$,
    \item $\rho_{AB} \in \ppt$,
    \item $\rho_{AB} \in \ut$,
    \item $\rho_{AB} \in \red$,
    \item $\rho_{AB} \in \maj$,
    \item $\rho_{AB} \in \cond$.
\end{enumerate}
\end{Prop}

This proposition implies that when $\rho_{AC}$ belongs to $\ppt$, 
if $\rho_{AB}$ is entangled, it must exhibit \emph{strong} entanglement, that is, 
it cannot belong to any class within the hierarchical implication chain in Eq.~(\ref{eq:hierarchy1}).  
This result provides the first formal realization of the CMoE principle, namely,  
if one side is weakly entangled, the other must be strongly entangled.  
In the next subsection, we review an alternative formulation of the CMoE.

\subsection{Converse monogamy of entanglement based on distillability}

The study by Singh and Datta~\cite{singh2023fully} did not explicitly employ the term, the CMoE, yet their results can be interpreted as a distinct form of it.  
Their approach emphasizes distillability rather than separability, 
offering a complementary perspective to that of Hayashi and Chen, as seen in the following two propositions.

\begin{Prop}\label{prop:SD1}
Let $\ket{\psi}_{ABC}$ be a tripartite pure state with the reduced state $\rho_{AC} \in \uo$. 
Then, for the reduced state $\rho_{AB}$, the following are equivalent:
\begin{enumerate}
    \item $\rho_{AB} \in \sep$,
    \item $\rho_{AB} \in \ppt$,
    \item $\rho_{AB} \in \ut$.
\end{enumerate}
\end{Prop}

\begin{Prop}\label{prop:SD2}
Let $\ket{\psi}_{ABC}$ be a tripartite pure state with the reduced state $\rho_{AC} \in \ut$. 
Then, for the reduced state $\rho_{AB}$, the following are equivalent:
\begin{enumerate}
    \item $\rho_{AB} \in \sep$,
    \item $\rho_{AB} \in \ppt$,
    \item $\rho_{AB} \in \ut$,
    \item $\rho_{AB} \in \uo$.
\end{enumerate}
\end{Prop}

Proposition~\ref{prop:SD1} implies that when $\rho_{AC}$ belongs to $\uo$, if $\rho_{AB}$ is entangled, it must be distillable.
In addition, Proposition~\ref{prop:SD2} tells us that under the assumption that $\rho_{AC} \in \ut$, if $\rho_{AB}$ is entangled, it is one-way distillable using only one-way LOCC.  
In summary, the hierarchy of $\rho_{AB}$ with respect to distillability collapses into a single class if $\rho_{AC}$ is (one-way) undistillable.

This result does not generalize the CMoE demonstrated by Hayashi and Chen; 
instead, it illustrates a partial collapse of the hierarchy under weaker assumptions.  
While belonging to $\ppt$ induces a complete collapse of all classes in the implication chains in Eq.~(\ref{eq:hierarchy1}), 
belonging to $\ut$ or $\uo$ produces only a restricted, distillability-oriented version of the CMoE.  
These observations set the stage for the next section, where we extend both results to their maximal forms.


\section{Maximal extensions of the Converse Monogamy of Entanglement}\label{sec:main}

In this section, we present two maximal extensions of the CMoE for a tripartite pure state $\ket{\psi}_{ABC}$. 
Our results unify and generalize those previously obtained by Hayashi and Chen~\cite{hayashi2011weaker} and by Singh and Datta~\cite{singh2023fully}. 
Before stating the main theorems, we recall several spectral relations that form the technical foundation of our proofs.

\subsection{Extension of the converse monogamy of entanglement}

The CMoE for a tripartite pure state can generally be understood through the complementary relations between its two bipartite reduced states.
Let $\rho_{AB} = \Tr_C(\ket{\psi}\!\bra{\psi}_{ABC})$ and $\rho_{AC} = \Tr_B(\ket{\psi}\!\bra{\psi}_{ABC})$.  
By the Schmidt decomposition, $\rho_{AC}$ and $\rho_B$ (or $\rho_{AB}$ and $\rho_C$) share the same nonzero spectrum.  
Therefore, comparing $\rho_{AB}$ and $\rho_{AC}$ is equivalent to comparing $\rho_{AB}$ and $\rho_B$, and this spectral correspondence constitutes the mathematical backbone of the CMoE structure.

We first recall several useful properties that hold under these spectral conditions.

\begin{Prop}[\cite{horodecki2000operational,chen2011distillability}]\label{prop::undsep}
Let $\rho_{AB} \in \ut$. If
\[
\rank(\rho_{AB}) = \max \{ \rank(\rho_A), \rank(\rho_B) \},
\] 
then $\rho_{AB} \in \sep$.
\end{Prop}

\begin{Prop}[\cite{hayashi2011weaker}]\label{prop:CH}
Let $\ket{\psi}_{ABC}$ be a pure state and suppose that its reduced state $\rho_{AC}$ belongs to $\sep$. Then
\[
\rho_{AB} \in \sep 
\quad \Longleftrightarrow \quad 
S(\rho_{AB}) = S(\rho_B).
\]
\end{Prop}

Building on these facts, we now present the extended forms of the CMoE for tripartite pure states.  
The following theorems encompass both the Hayashi–Chen and Singh–Datta results, and moreover demonstrate that the CMoE remains valid under weaker assumptions.

\begin{Thm}\label{thm:full_extension}
Let $\ket{\psi}_{ABC}$ be a tripartite pure state with $\rho_{AC} \in \ut$.  
Then, for the reduced state $\rho_{AB}$, the following statements are equivalent:
\begin{enumerate}
    \item $\rho_{AB} \in \sep$,
    \item $\rho_{AB} \in \ppt$,
    \item $\rho_{AB} \in \ut$, 
    \item $\rho_{AB} \in \uo$,   
    \item $\rho_{AB} \in \red$,    
    \item $\rho_{AB} \in \maj$,
    \item $\rho_{AB} \in \cond$,
    \item $\rho_{AB} \in \condrr$.
\end{enumerate}
\end{Thm}

\begin{proof}
It suffices to show that $\rho_{AB} \in \sep$ if it belongs to $\condrr$. 
Thus, let us assume that $\rho_{AB} \in\condrr$. Then
since $\rho_{AC} \in \ut$, we have $\rank(\rho_{AC}) \ge \max \{ \rank(\rho_A), \rank(\rho_C) \}$~\cite{horodecki2003rank}, and since $\rho_{AC}$ is also contained in both $\maj$ and $\cond $, we obtain $\rho_C \succ \rho_{AC}$ and $S(\rho_C) \le S(\rho_{AC})$. 
In addition, from the definition of $\rho_{AB} \in \condrr$, we get $S(\rho_{AB}) \ge S(\rho_B)$, which implies $S(\rho_C) \ge S(\rho_{AC})$.  
Hence $S(\rho_C) = S(\rho_{AC})$, and by the strict Schur concavity of the von Neumann entropy~\cite{marshall1979inequalities,hanson2018strict}, $\rho_C$ and $\rho_{AC}$ are isospectral.  
Therefore,
\[
\rank(\rho_{AC}) = \rank(\rho_{C}) = \max \{ \rank(\rho_A), \rank(\rho_C) \}.
\]
By Proposition~\ref{prop::undsep}, we have $\rho_{AC} \in \sep$, and thus, by Proposition~\ref{prop:CH}, $\rho_{AB} \in \sep$.
\end{proof}

\begin{Thm}\label{thm:full_equivalence}
Let $\ket{\psi}_{ABC}$ be a tripartite pure state with $\rho_{AC} \in \condrr$.  
Then, for the reduced state $\rho_{AB}$, the following statements are equivalent:
\begin{enumerate}
    \item $\rho_{AB} \in \sep$,
    \item $\rho_{AB} \in \ppt$,
    \item $\rho_{AB} \in \ut$.
\end{enumerate}
\end{Thm}

\begin{proof}
It suffices to show that $\rho_{AB} \in \sep$ if $\rho_{AB} \in \ut$.  
Assume that $\rho_{AB} \in \ut$. Then  since $\ut\subset\condrr$, we have 
$S(\rho_B) \le S(\rho_{AB})$.  
Combining this with the definition of $\rho_{AC} \in \condrr$, $S(\rho_C) \le S(\rho_{AC})$, and thus we obtain $S(\rho_B) = S(\rho_{AB})$, which means that $\rho_B$ and $\rho_{AB}$ are isospectral since $\rho_{AB}\in \ut \subset \maj$ and the von Neumann entropy is strictly Schur concave, 
as in the proof of Theorem~\ref{thm:full_extension}. 
Hence we obtain $\rank(\rho_{AB})=\rank(\rho_B)$. 
In addition, since $\rho_A \succ \rho_{AB}$, it follows that $\rank(\rho_A)\le\rank(\rho_{AB})$~\cite{watrous11}.
Therefore, by Proposition~\ref{prop::undsep}, $\rho_{AB}$ must be separable.
\end{proof}

Theorem~\ref{thm:full_extension} and Theorem~\ref{thm:full_equivalence} directly imply the following corollary.
\begin{Cor}\label{cor:summary}
Let $|\psi\rangle_{ABC}$ be a pure state such that both $\rho_{AC}$ and $\rho_{BC}$ are entangled. Then
\begin{enumerate}[label=\textnormal{(\roman*)}]
\item If $\rho_{AB}\in\condrr$, then $\rho_{AC}$ and $\rho_{BC}$ are distillable.
\item If $\rho_{AB}\in\ut$, then $\rho_{AC}$ and $\rho_{BC}$ are not contained in $\condrr$.
\end{enumerate}
\end{Cor}

We now remark that since both $\cond$ and $\uo$ are included in $\condrr$ and $\ppt \subset \ut$, Theorem~\ref{thm:full_extension} and Theorem~\ref{thm:full_equivalence} can be regarded as  generalizations of both Proposition~\ref{thm:HC} and Proposition~\ref{prop:SD2}, and Proposition~\ref{prop:SD1}, respectively. 
Thus our two main theorems together imply that if one reduced bipartite state of a tripartite pure state does not exhibit a hierarchical collapse within the hierarchy, its complementary bipartite state must be distillable or one-way distillable due to the hashing bound, which means that the distillable entanglement of a bipartite state $\rho_{AB}$ is lower-bounded by $\max\{0,-S(A|B)_{\rho_{AB}}\}$, 
where $S(A|B)=S(AB)-S(B)$ is the quantum conditional entropy. 

Furthermore, even though $\rho_{AC}$ does not satisfy the assumption in Proposition~\ref{thm:HC}, Proposition~\ref{prop:SD1}, or Proposition~\ref{prop:SD2}, that is, $\rho_{AC}$ has a negative eigenvalue after applying the partial transpose, it is one-way distillable, or it is distillable, our main theorems ensure that $\ket{\psi}_{ABC}$ still satisfies the CMoE, under the weaker assumption in the theorems.
The following example explicitly illustrates this scenario.

\begin{Ex}\label{ex:cond_D1}
Consider the one-way distillable state introduced by Horodecki \textit{et al.}~\cite{horodecki2005locking}: 
let $\rho_{AC}$ be a state given by 
\[
\rho_{AC} = \rho_{A_1A_2C}
= \frac{1}{2}\ket{0}\bra{0}_{A_1}\otimes\ket{\psi^+}\bra{\psi^+}_{A_2C}
+ \frac{1}{2}\ket{1}\bra{1}_{A_1}\otimes\frac{I_{A_2C}}{d^2},
\]
where $A=A_1A_2$ and 
\begin{equation}
\ket{\psi^+}=\frac{1}{\sqrt{d}}\sum_{i=0}^{d-1}\ket{ii}.
\label{eq:psi}
\end{equation}

Then from an explicit one-way distillation protocol, we can see that $\rho_{AC}$ is one-way distillable, i.e., $\rho_{AC}\notin \uo$~\cite{horodecki2005locking}.
A direct calculation gives
\begin{equation}
S(\rho_A) = S(\rho_{AC}) = 1 + \log d,
\quad
S(\rho_{C}) =  \log d,
\end{equation}
since $\rho_A = \frac{1}{2d}I_{A}$ and $\rho_C = \frac{1}{d}I_{C}$.
So, 
$S(\rho_{AC}) = S(\rho_{A})$ and $S(\rho_{AC}) > S(\rho_{C})$,
which implies $\rho_{AC}\in \condrr\setminus\maj$ 
from the fact that $\rho_A \nsucc \rho_{AC}$ and $\rho_C \nsucc \rho_{AC}$.

Let $\rho_{AB}$ be the complementary state of $\rho_{AC}$, i.e., 
\[
\rho_{AB}=\Tr_C(\ket{\psi}\bra{\psi}_{ABC}),
\]
where $\ket{\psi}_{ABC}$ is a purification of $\rho_{AC}$.
Then since $\rank(\rho_{AB})=\rank(\rho_C)=d$ and $\rank(\rho_{B})=d^2+1$, 
\[\rank(\rho_{AB}) < \max\{\rank(\rho_A),\rank(\rho_B)\},\] and hence
$\rho_{AB}$ is distillable~\cite{horodecki2003rank}.
\end{Ex}
This example confirms that even when $\rho_{AC}$ is one-way distillable, the CMoE remains valid for $\ket{\psi}_{ABC}$ if it is in $\condrr$, as in Theorem~\ref{thm:full_equivalence}. 
Hence, by exchanging the systems $B$ and $C$, 
this state also satisfies the CMoE as represented by Theorem~\ref{thm:full_extension}, but it cannot be explained by the previous results about the CMoE.

Finally, Corollary~2.13 of Ref.~\cite{singh2023fully} implies a similar structural relation between two bipartite reduced states of a tripartite pure state.  
By combining Theorems~\ref{thm:full_extension} and~\ref{thm:full_equivalence}, 
we obtain the following \textit{maximally extended} form that contains their result.

\begin{Cor}\label{thm:summary2}
For a tripartite pure state $\ket{\psi}_{ABC}$, the following statements are equivalent:
\begin{enumerate}
    \item $\rho_{AB}$, $\rho_{AC} \in \sep$,
    \item $\rho_{AB}$, $\rho_{AC} \in \ppt$,
    \item $\rho_{AB}$, $\rho_{AC} \in \ut$,
    \item $\rho_{AB} \in \ut$, $\rho_{AC} \in \uo$ (or vice versa),
    \item $\rho_{AB} \in \ut$, $\rho_{AC} \in \condrr$ (or vice versa).
\end{enumerate}
\end{Cor}

\subsection{Applications of our CMoE} 
We now present two corollaries as direct applications of Theorem~\ref{thm:full_equivalence} to bound entanglement and quantum capacities for quantum channels, respectively. The first one related to bound entanglement is as follows.

\begin{Cor}
If $\rho_{AB}$ is bound entangled, i.e., $\rho_{AB} \in \ut \setminus \sep$, then its complementary state $\rho_{AC}$ has negative quantum conditional entropy (or equivalently, $\rho_{AB}$ has positive quantum conditional entropy), i.e., $S(A|C)_{\rho_{AC}}<0$ (or equivalently, $S(A|B)_{\rho_{AB}}>0$).
\end{Cor}

For a quantum channel $\mathcal{E}$ from a quantum system $A$ to a quantum system $B$, let the one-shot quantum capacity $Q^{(1)}(\mathcal{E})$ and the asymptotic quantum capacity $Q(\mathcal{E})$ be defined by
\begin{align}
Q^{(1)}(\mathcal{E}) &= \max_{\rho}\left[S(\mathcal{E}(\rho)) - S(\mathcal{E}^c(\rho))\right], 
\label{eq:Q1}\\
Q(\mathcal{E}) &= \lim_{n\rightarrow\infty}\frac{Q^{(1)}(\mathcal{E}^{\otimes n})}{n},
\label{eq:Q}
\end{align}
where $\mathcal{E}^c$ is the complementary channel from the quantum system $A$ to the environment system $C$ of $\mathcal{E}$, i.e., $\mathcal{E}^c(\rho)=\Tr_B(U\rho U^\dagger)$ for any state $\rho$ when $\mathcal{E}(\rho)=\Tr_C(U\rho U^\dagger)$ for some isometry $U$ from $A$ to $BC$ by Stinespring dilation theorem.  
Then it follows that for any quantum channel $\mathcal{E}$, 
\[
Q(\mathcal{E}) \ge Q^{(1)}(\mathcal{E}) \ge I^c_\rightarrow(J_{AB}(\mathcal{E})),
\]
where $J_{AB}(\mathcal{E})$ denotes the Choi state of $\mathcal{E}$, i.e., $J_{AB}(\mathcal{E})=(\id\otimes\mathcal{E})(\ket{\psi^+}\bra{\psi^+})$, 
where $\ket{\psi^+}$ is the maximally entangled state in Eq.~(\ref{eq:psi}) and $d=\dim(A)$, and $I^c_\rightarrow$ denotes the coherent information, which equals the negation of the quantum conditional entropy, i.e., $I^c_\rightarrow(\rho_{AB})=-S(A|B)_{\rho_{AB}}$.  
Theorem~\ref{thm:full_equivalence} then yields the following result concerning quantum capacities.

\begin{Cor}
Let $\mathcal{E}$ be a quantum channel whose complementary channel has zero quantum capacity (or zero one-shot quantum capacity), that is, $Q(\mathcal{E}^c)=0$ (or $Q^{(1)}(\mathcal{E}^c)=0$).  
Then $\mathcal{E}$ cannot be entanglement-binding; in other words, if $\mathcal{E}$ is a PPT channel, it must be entanglement-breaking.  
Equivalently, if the complementary channel of $\mathcal{E}$ is entanglement-binding, then $\mathcal{E}$ possesses strictly positive quantum and one-shot quantum capacities, that is, $Q(\mathcal{E})\ge Q^{(1)}(\mathcal{E})>0$.
\end{Cor}

\subsection{Our CMoE is maximal}

In the previous subsection, we hierarchically formulated two generalizations of the CMoE.  
A natural question may here arise: how far can such generalizations be extended?  
We now show that the proposed extensions are maximal within the hierarchy we consider in this work.

Within this hierarchy, Theorem~\ref{thm:full_extension} would be further extended if it remained valid when the assumption on $\rho_{AC}$ were replaced with $\rho_{AC} \in \red$ or $\rho_{AC} \in \uo$.  
Similarly, Theorem~\ref{thm:full_equivalence} would extend further if $\rho_{AB} \in \red$ or $\rho_{AB} \in \uo$ were added in the equivalence conditions for $\rho_{AB}$ of Theorem~\ref{thm:full_equivalence}.  
However, an explicit counterexample shows that such extensions do not hold, as follows.

Consider the following tripartite pure state: let $\ket{\psi}_{ABC}$ be a state given by
\begin{eqnarray*}
\ket{\psi}_{ABC}
&=&\frac{1}{\sqrt{6}}\left(
\ket{012}-\ket{102}-\ket{021}+\ket{201}+\ket{120}-\ket{210}
\right) \\
&=&\frac{1}{\sqrt{3}}\left(\ket{0}\ket{\psi_{12}^-}+\ket{1}\ket{\psi_{20}^-}+\ket{2}\ket{\psi_{01}^-}\right),
\end{eqnarray*}
where \[\ket{\psi_{ij}^-}=\frac{1}{\sqrt{2}}\left(\ket{ij}-\ket{ji}\right).\]
In this case, $\rho_{AB} = \rho_{AC} \in \uo$~\cite{cubitt2008structure}. 
Moreover, since the maximal eigenvalue of $\rho_{AB}$ is $1/3$,  
\begin{equation}
\rho_A \otimes I_B - \rho_{AB} = I_A \otimes \rho_B - \rho_{AB} = \frac{1}{3} I_9 - \rho_{AB} \ge 0,
\end{equation}
while $\rho_{AB}^\Gamma$ has a negative eigenvalue $-\frac{1}{3}$, implying $\rho_{AB} \notin \ppt$.  
Since $\rho_{AB}$, $\rho_{A}$, and $\rho_{B}$ are isospectral, $\rank(\rho_{AB}) = \max \{ \rank(\rho_A), \rank(\rho_B) \}$, and Proposition~\ref{prop::undsep} hence ensures that $\rho_{AB}$ is distillable.  
Therefore, both $\rho_{AB}$ and $\rho_{AC}$ lie in $\uo \setminus \ut$, indicating that Theorems~\ref{thm:full_extension} and~\ref{thm:full_equivalence} cannot be extended further.   
This is due to the fact that 
when $\rho_{AC} \notin \ut$, the total hierarchical collapse in Theorem~\ref{thm:full_extension} does not need to occur, and   
although $\rho_{AC} \in \condrr$, the hierarchy combining that in Theorem~\ref{thm:full_equivalence} with $\rho_{AB}\in \uo$ or $\rho_{AB}\in \red$ does not exhibit a collapse.
We summarize our findings as follows.
\begin{Thm}[Maximal extension of CMoE]\label{thm:summary}
Within the hierarchy in Fig.~(\ref{Fig:Hierarchy}), neither Theorem~\ref{thm:full_extension} nor Theorem~\ref{thm:full_equivalence} admits any weaker implications.
\end{Thm}


\section{Conclusion}\label{sec:con}

The distribution of quantum entanglement is constrained by the monogamy of entanglement, whereas its converse relation, the CMoE, can be understood more precisely through the hierarchy formed by various separability criteria. 
From this perspective, we revisited the two previously known forms of the CMoE and characterized the maximal extensions attainable within this hierarchical structure. 
Consequently, the CMoE relations are unified within a single coherent framework, and each extension is shown to be a maximal form that cannot be further generalized.

This result clarifies that the CMoE is not merely a comparison between individual separability conditions, but rather a mathematical structure that describes how far the collapse can occur within the hierarchy of bipartite entanglement (from separability to undistillability). 
In particular, we confirmed that when one of the bipartite reduced states is in $\ut$ and another one is in $\condrr$, the complementary state necessarily becomes $\sep$. 
This observation reveals that the result of Singh and Datta emerges as a special instance of our maximal extension. 
Such a hierarchical interpretation provides a framework for viewing the CMoE as a manifestation of the complementarity between strong and weak entanglement.

Meanwhile, the hierarchy of separability criteria underlying the definition of the CMoE bifurcates into two principal branches, $\red$ and $\uo$, 
when it comes to the classes containing $\ut$. 
Although these two sets are clearly distinct, their exact inclusion relation remains unresolved. 
If it were proven that $\red \subset \uo$, the entire hierarchical structure of the CMoE could be formulated in a simpler and more transparent form. 
Determining this relationship is therefore essential for a complete understanding of the maximal extensibility of the CMoE, and it may offer valuable insight into the asymmetry of entanglement distribution and the origin of informational complementarity.

\section*{ACKNOWLEDGMENTS}
We thank Hyukjoon Kwon and Ray Ganardi for discussions.
This research was supported by
the Institute of Information \& Communications Technology Planning \& Evaluation (IITP) grant funded by the Ministry of Science and ICT (MSIT) (No. RS-2025-02304540). 
S.L. acknowledges support from the National Research Foundation of Korea (NRF) of Korea grants funded by the MSIT  (No. RS-2024-00432214 and No. RS-2022-NR068791) and Creation of the Quantum Information Science R\&D Ecosystem (No. RS-2023-NR068116) through the NRF funded by the MSIT.

\bibliography{full}

\end{document}